\documentclass[runningheads, envcountsame, a4paper]{llncs}
\usepackage[pdftex]{graphicx}
\usepackage{amsmath}
\usepackage[unicode, colorlinks, urlcolor=blue]{hyperref}

\usepackage[linesnumbered,ruled,noend]{algorithm2e}

\usepackage[margin=2cm]{geometry}
\usepackage{amssymb}
\usepackage{tikz}
\usetikzlibrary{shapes,calc,math,backgrounds}
\usepackage{enumerate}
\usepackage{adjustbox}

\spnewtheorem{claim}{Claim}{\itshape}{\rmfamily}

\begin{document}
\title{TS-Reconfiguration of $k$-Path Vertex Covers\\ in Caterpillars for $k \geq 4$}
\author{Duc~A.~Hoang\orcidID{0000-0002-8635-8462}
}
\authorrunning{D.A.~Hoang}
\institute{Graduate School of Informatics, Kyoto University, Japan\\
	\email{hoang.duc.8r@kyoto-u.ac.jp}
}
\maketitle              %
\begin{abstract}
A $k$-path vertex cover ($k$-PVC) of a graph $G$ is a vertex subset $I$ such that each path on $k$ vertices in $G$ contains at least one member of $I$.
Imagine that a token is placed on each vertex of a $k$-PVC.
Given two $k$-PVCs $I, J$ of a graph $G$, the \textsc{$k$-Path Vertex Cover Reconfiguration ($k$-PVCR)} under Token Sliding ($\mathsf{TS}$) problem asks if there is a sequence of $k$-PVCs between $I$ and $J$ where each intermediate member is obtained from its predecessor by sliding a token from some vertex to one of its unoccupied neighbors.
This problem is known to be $\mathtt{PSPACE}$-complete even for planar graphs of maximum degree $3$ and bounded treewidth and can be solved in polynomial time for paths and cycles.
Its complexity for trees remains unknown.
In this paper, as a first step toward answering this question, for $k \geq 4$, we present a polynomial-time algorithm that solves \textsc{$k$-PVCR} under $\mathsf{TS}$ for caterpillars (i.e., trees formed by attaching leaves to a path).

\end{abstract}

\section{Introduction}
\label{sec:introduction}

Recently, \textit{reconfiguration problems} have attracted the attention from both theoretical and practical viewpoints.
The input of a reconfiguration problem consists of two \textit{feasible solutions} of some computational problem (e.g., \textsc{Satisfiability}, \textsc{Independent Set}, \textsc{Vertex Cover}, \textsc{Dominating Set}, etc.) and a \emph{reconfiguration rule} that describes an adjacency relation between solutions.
One of the primary goal is to decide whether one feasible solution can be transformed into the other via a sequence of adjacent feasible solutions where each intermediate member is obtained from its predecessor by applying the given reconfiguration rule exactly once.
Such a sequence, if exists, is called a \emph{reconfiguration sequence}.
Readers may recall the classic Rubik's cube puzzle as an example of a reconfiguration problem, where each configuration of the Rubik's cube corresponds to a feasible solution, and two configurations (solutions) are adjacent if one can be obtained from the other by rotating a face of the cube by either $90$, $180$, or $270$ degree.
The question is whether one can transform an arbitrary configuration to the one where each face of the cube has only one color.
For an overview of this research area, see the recent surveys~\cite{Heuvel13,Nishimura18,MynhardtN19}.

Let $G = (V, E)$ be a simple graph.
Let $k \geq 2$ be a fixed positive integer.
A subset $I$ of $V$ is called a \emph{$k$-path vertex cover ($k$-PVC)} if every path on $k$ vertices in $G$ contains at least one vertex from $I$.
The \textsc{$k$-Path Vertex Cover} problem asks, given a graph $G$ and an integer $s$, if there is a $k$-path vertex cover of $G$ whose size is at most some positive integer $s$.
Motivated by the importance of a problem related to secure communication in wireless sensor networks, Bre{\v{s}}ar et al. initiated the study of \textsc{$k$-PVC} in~\cite{BrevsarKKS11} (as a generalisation of the well-known \emph{vertex cover}).
It is known that \textsc{$k$-Path Vertex Cover} is $\mathtt{NP}$-complete for every $k \geq 2$~\cite{AcharyaCBG12,BrevsarKKS11}.
Subsequent work regarding the \emph{maximum} variant~\cite{MiyanoSUYZ18} and \emph{weighted} variant~\cite{BrevsarKSS14} of \textsc{$k$-Path Vertex Cover} has also been considered in the literature.
Recently, the study of \textsc{$k$-Path Vertex Cover} and related problems has gained a lot of attraction from both theoretical aspect~\cite{RanZHLD19,Tsur19} %
and practical application~\cite{FunkeNS14}. %
For more details, see the survey~\cite{Tu22}.

Imagine that each vertex of $G$ contains at most one token and the set of tokens form a $k$-PVC.
Two $k$-PVCs $I, J$ of a graph $G$ are \textit{adjacent} under \textit{Token Sliding} ($\mathsf{TS}$) if there are two vertices $u, v \in V(G)$ such that $I$ is obtained from $J$ (or vice versa) by sliding a token along the edge $uv \in E(G)$ from $u$ to $v$ (or from $v$ to $u$).
The \textsc{$k$-Path Vertex Cover Reconfiguration ($k$-PVCR)} problem under $\mathsf{TS}$ asks whether there is a sequence of adjacent $k$-PVCs that transforms $I$ into $J$.
When $k = 2$, it is simply called the \textsc{Vertex Cover Reconfiguration (VCR)} problem.
The \textsc{VCR} problem has been very well-studied in the literature under $\mathsf{TS}$ as well as other reconfiguration rules.
Readers are referred to~\cite{Nishimura18} for a quick summary of the known results.

Hoang et al.~\cite{HoangSY22} initiated the study of \textsc{$k$-PVCR} under $\mathsf{TS}$ and other reconfiguration rules (namely \textit{Token Jumping ($\mathsf{TJ}$)} and \textit{Token Addition/Removal ($\mathsf{TAR}(u)$)} which respectively involve moving a token to any unoccupied vertex and adding/removing a token such that the resulting set always contains at most $u$ tokens for some given positive integer $u$) for fixed $k \geq 3$.
In particular, they showed that \textsc{$k$-PVCR} under $\mathsf{TS}$ remains $\mathtt{PSPACE}$-complete for planar graphs of maximum degree $3$ and bounded bandwidth, chordal graphs, and perfect graphs.
We note that in a preliminary version~\cite{HoangSY20} of the mentioned paper, the authors wrongly claimed that \textsc{$k$-PVCR} under $\mathsf{TS}$ is $\mathtt{PSPACE}$-complete for bipartite graphs.
The reason is that they reduced from the \textsc{Minimum Vertex Cover Reconfiguration} problem---a problem that has \textit{not} yet been shown to be $\mathtt{PSPACE}$-hard for bipartite graphs under $\mathsf{TS}$.
On the other hand, it is not hard to verify that their reduction remains true for perfect graphs---a superclass of bipartite graphs.
On the positive side, they designed polynomial-time algorithms for solving \textsc{$k$-PVCR} under $\mathsf{TS}$ for some very restricted graphs, namely paths and cycles.
Unfortunately, the complexity of \textsc{$k$-PVCR} under $\mathsf{TS}$ for trees remains unknown, while for $\mathsf{TJ}/\mathsf{TAR}(u)$ it can be solved in linear time~\cite{HoangSY22}.
As an attempt to tackle this open question, in this paper, when $k \geq 4$, we present a polynomial-time algorithm for solving this problem with caterpillars---a subclass of trees---as the input graph.
Roughly speaking, a \textit{caterpillar} is obtained by attaching leaves (i.e., vertices of degree $1$) to a given path (called the \textit{spine}, or \textit{backbone} path).
Intuitively, by characterizing the tokens that ``cannot be moved at all'' in polynomial time, we show that one can efficiently identify all no-instances of \textsc{$k$-PVCR} under $\mathsf{TS}$.
Though this is a familiar technique in designing polynomial-time algorithms for solving several reconfiguration problems, our characterization of such tokens, in several cases, is non-trivial.
In particular, it is done via studying ``a special region $\mathcal{H}$ surrounding a token $t$ satisfying that $t$ can be moved only if no token outside $\mathcal{H}$ can be moved in''. 
(For more details, see Section~\ref{sec:rigid-tokens}.)
Additionally, in a yes-instance, we explicitly describe how to construct a sequence of $\mathsf{TS}$-moves that transforms one $k$-PVC into another.

\section{Preliminaries}

In this section, we define some notation and terminology.
We also briefly recall an useful algorithm from~\cite{HoangSY22} (Algorithm~\ref{algo:partition}).

\subsection{Graph Notation}
\label{sec:graph-notation}

Readers are referred to~\cite{Diestel2017} for the concepts that are not mentioned here.
We use $V(G)$ and $E(G)$ to denote the sets of vertices and edges of a (simple, undirected) graph $G$, respectively.
For $v \in V(G)$, we use $N_G(v)$ to indicate the set of $v$'s \textit{neighbors} in $G$, i.e., the set of vertices that are adjacent to $v$.
The \textit{closed neighborhood} of $v$ in $G$, denoted by $N_G[v]$, is simply the set $N_G(v) \cup \{v\}$.
The \textit{degree} of $v$ in $G$, denoted by $\deg_G(v)$, is the number of vertices in $G$ which are adjacent to $v$, i.e., $\deg_G(v) = \vert N_G(v) \vert$.
For $u, v \in V(G)$, the \textit{distance} between $u$ and $v$ in $G$, denoted by $\text{dist}_G(u, v)$, is the length of a path in $G$ between them whose number of edges is smallest.

$H$ is a \textit{proper subgraph} of $G$ if it is a subgraph of $G$ and $V(G) \setminus V(H) \neq \emptyset$.
For a subset $X$ of $V(G)$, we denote by $G - X$ the graph obtained from $G$ by removing vertices in $X$ (and their incident edges), and $G[X]$ the subgraph induced by vertices in $X$.
For convenience, if $X$ has exactly one member, say $v$, then we write $G - v$ to indicate $G - \{v\}$.
Additionally, if $X = V(H)$ for some induced subgraph $H$ of $G$, then we write $G - H$ to indicate $G - V(H)$.
We say that two vertices $u$ and $v$ are \textit{in the same component} of $G$ if there is a path in $G$ between $u$ and $v$.
A \textit{$k$-path} is simply a path on $k$ vertices.
A vertex $v$ \textit{covers} a path $P$ if $v \in V(P)$.
A \textit{$k$-path vertex cover ($k$-PVC)} $I$ of $G$ is a vertex subset such that every $k$-path is covered by some member of $I$.
We denote by $\psi_k(G)$ the \textit{minimum} size of a $k$-PVC of $G$.

A \textit{tree} is a simple, undirected, connected graph that contains no cycles.
For any pair of vertices $u, v$ in a tree $T$, we use $P_{uv}$ to denote the unique $uv$-path connecting $u$ and $v$ in $T$.
A \textit{caterpillar} $G$ is a tree where $V(G)$ can be partitioned into two sets $S$ (\textit{spine}) and $L$ (\textit{leaves}) satisfying that (a) vertices in $S$ induce a path $s_1s_2\dots s_\ell$ in $G$, and (b) each vertex in $L$ has degree $1$, and its unique neighbor is in $S$.
We use $G = (S \cup L, E)$ to indicate a caterpillar with spine $S$ and leaves $L$.
To avoid ambiguity, we assume that $\ell \geq 2$, $\deg_G(s_1) \geq 2$, and $\deg_G(s_\ell) \geq 2$.
For a vertex $s = s_i \in S$ ($1 \leq i \leq \ell$), we denote by $l(s_i) = s_{i-1}$ (resp. $r(s_i) = s_{i+1}$) the \textit{left-neighbor} (resp. \textit{right-neighbor}) of $s$.
In order to make the concept of left/right-neighbors well-defined, we set $s_0 = s_1$ and $s_{\ell+1} = s_\ell$.
We also denote by $L_G(s)$ the set $N_G(s) \cap L$ containing \textit{leaf-neighbors} (i.e., degree-$1$ neighbors) of $s$ in $G$ and by $L_G[s]$ the set $L_G(s) \cup \{s\}$.

\subsection{Reconfiguration Notation}
\label{sec:reconf-notation}

We denote by $(G, I, J)$ an instance of \textsc{$k$-PVCR} where $I, J$ are $k$-PVCs of $G$.
Imagine that a token is placed on each vertex in a $k$-PVC of a graph $G$. 
A \textit{$\mathsf{TS}$-sequence} in $G$ between $I$ and $J$ is the sequence $\langle I = I_0, I_1, \dots, I_q = J \rangle$ such that for $i \in \{0, \dots, q-1\}$, the set $I_i$ is a $k$-PVC of $G$ and there exists a pair $x_i, y_i \in V(G)$ such that $I_i \setminus I_{i+1} = \{x_i\}$, $I_{i+1} \setminus I_i = \{y_i\}$, and $x_iy_i \in E(G)$.
In other words, $I_{i+1}$ is obtained from $I_i$ by \textit{immediately sliding} a token from $x_i$ to $y_i$ along the edge $x_iy_i$.
In short, $\mathcal{S}$ can be viewed as a (ordered) sequence of either $k$-PVCs or token-slides.
With respect to the latter viewpoint, we say that $\mathcal{S}$ \textit{slides/moves a token $t$ from $u$ to $v$ in $G$} if $t$ is originally placed on $u \in I_0$ and finally on $v \in I_q$ after performing $\mathcal{S}$. 

For a $\mathsf{TS}$-sequence $\mathcal{S} = \langle I_0, I_1, \dots, I_q \rangle$, we denote by $\text{rev}(\mathcal{S})$ the \emph{reverse} of $S$, i.e., the $\mathsf{TS}$-sequence $\langle I_q, \dots, I_1, I_0\rangle$.
For two $\mathsf{TS}$-sequences $\mathcal{S} = \langle I_0, I_1, \dots, I_p \rangle$ and $\mathcal{S}^\prime = \langle I^\prime_0, I^\prime_1, \dots, I^\prime_q \rangle$, if $I_p = I^\prime_0$ then we say that they can be \emph{concatenated} and define their \emph{concatenation} $\mathcal{S} \oplus \mathcal{S}^\prime$ as the $\mathsf{TS}$-sequence $\langle I_0, I_1, \dots, I_p, I^\prime_1, \dots, I^\prime_q \rangle$.
We assume for convenience that if $\mathcal{S}^\prime$ is empty then $\mathcal{S} \oplus \mathcal{S}^\prime = \mathcal{S}^\prime \oplus \mathcal{S} = \mathcal{S}$.

For a $k$-PVC $I$ of $G$, we say that a token $t$ on $u \in I$ is \textit{$(G, I)$-rigid} if it cannot be moved at all, that is, for any $k$-PVC $J$ of $G$ obtained from $I$ via a $\mathsf{TS}$-sequence in $G$, we always have $u \in J$.
In other words, there is no $\mathsf{TS}$-sequence that slides $t$ from $u$ to any of its neighbors in $G$.
If $t$ is \textit{not} $(G, I)$-rigid, we call it a \textit{$(G, I)$-movable} token.
We denote by $\mathcal{R}(G, I)$ the set of all vertices in $G$ where $(G, I)$-rigid tokens are placed.

\subsection{Partitioning Trees}
\label{sec:partitioning-trees}

In this section, we describe a slightly modified version of an algorithm used in~\cite{HoangSY22} for partitioning a given tree $T$.
We remark that this algorithm is also crucial for solving the problem under $\mathsf{TJ}$ and $\mathsf{TAR}$.
A \emph{properly rooted subtree} $T_v$ of a $n$-vertex rooted tree $T$ is a subtree of $T$ induced by the vertex $v$ and all its descendants (with respect to the root $r$) satisfying the following conditions
\begin{enumerate}[(1)]
	\item $T_v$ contains a $k$-path;
	\item $T_v - v$ does not contain a $k$-path.
\end{enumerate} 
The $O(n)$-time algorithm $\mathtt{Partition}(T, k, r)$ (Algorithm~\ref{algo:partition}) takes a tree $T$, an integer $k \geq 3$, and a root vertex $r \in V(T)$ as an input, and returns a minimum $k$-PVC $I(T, k, r)$ and a partition $P(T, k, r)$ of $T$.
In short, it systematically searches for a properly rooted tree $T_v$, decides whether $T_v$ is properly rooted, and if so, add $T_v$ to $P(T, k, r)$ and $v$ to $I(T, k, r)$, and remove $T_v$ from the input tree $T$. 
To check if $T$ contains a properly rooted subtree $T_v$, one can start by assigning $v$ to a vertex of largest \textit{depth} (i.e., distance from $r$) and verify if $T_v$ is properly rooted. If so, the answer is ``yes''. Otherwise, we assign $v$ to its parent and repeat, until a $T_v$ is found (answering ``yes'') or there is nothing to check (answering ``no'').
The number of subtrees in $P(T, k, r)$ is indeed $\psi_k(T)$---the minimum size of a $k$-PVC of $T$.

\begin{algorithm}[!ht]
	\KwIn{A tree $T$ on $n$ vertices rooted at $r$ and an integer $k \geq 3$.}
	\KwOut{A partition $P(T, k, r)$ and a minimum $k$-PVC $I(T, k, r)$ of $T$.}
	\SetArgSty{textbb}   
	$i := 1$\;
	\While{$T$ contains a properly rooted subtree $T_v$ }
	{            
		\If{$T - T_v$ contains a properly rooted subtree}
		{
			$T_i(r) := T_v$; $v_i(r) := v$\;
			$i := i + 1$\;
		}
		\Else{
			$T_i(r) := T$; $v_i(r) := v$\;
		}
		$T := T - T_v$\;
	}
	$P(T, k, r) = \{T_1(r), \dots, T_i(r)\}$\;
	$I(T, k, r) = \{v_1(r), \dots, v_i(r)\}$\;
	\Return $P(T, k, r)$ and $I(T, k, r)$\;
	\caption{$\mathtt{Partition}(T, k, r)$.}
	\label{algo:partition}
\end{algorithm}

\section{Rigid Tokens}
\label{sec:rigid-tokens}

Let $I$ be a $k$-PVC of a $n$-vertex caterpillar $G = (S \cup L, E)$.
In this section, we will first characterize whether a token $u \in I$ is $(G, I)$-rigid.
Using our characterization, one can design a polynomial-time algorithm to find all such tokens when $k \geq 4$.
The following lemma is straightforward.

\begin{lemma}\label{lem:rigid-tokens-in-L}
	Let $I$ be a $k$-PVC ($k \geq 3$) of a $n$-vertex caterpillar $G = (S \cup L, E)$, and let $u \in I$. 
	\begin{enumerate}[(a)]
		\item If $N_G(u) = \emptyset$, the token $t$ on $u$ is always $(G, I)$-rigid.
		\item If $u \in L$, the token $t$ on $u$ is $(G, I)$-rigid if and only if its unique neighbor $v$ satisfies $v \in I$ and the token $t_v$ on $v$ is $(G - u, I \cap V(G - u))$-rigid. 
	\end{enumerate}
\end{lemma}
\begin{proof}
	\begin{enumerate}[(a)]
			\item Trivial.
			
			\item The if direction is trivial.
			To show the only-if direction, it is sufficient to show that if either $v \notin I$ or $v \in I$ and the token $t_v$ on $v$ is $(G - u, I \cap V(G - u))$-movable, $t$ can be moved in $G$.
			Indeed, if $v \notin I$, we simply slide $t$ from $u$ to $v$, because any $k$-path covered by $u$ is also covered by $v$.
			On the other hand, if a token $t_v$ is placed on $v \in I$, since $t_v$ is $(G - u, I \cap V(G - u))$-movable, there is a $\mathsf{TS}$-sequence in $G - u$ (which is also in $G$) that moves $t_v$ to one of $v$'s neighbors. 
			Simply apply such a sequence, and then we can move $t$ from $u$ to $v$.
		\end{enumerate}
\qed\end{proof}

Consequently, it remains to characterize whether a token on $u \in I \cap S$ is $(G, I)$-rigid.
Additionally, by Lemma~\ref{lem:rigid-tokens-in-L}(a), it suffices to assume from this point forward that $N_G(u) \neq \emptyset$.

\begin{definition}\label{def:HGIu}
	Let $I$ be a $k$-PVC ($k \geq 3$) of a caterpillar $G = (S \cup L, E)$.
	For $u \in I \cap S$, we define $\mathcal{H}(G, I, u)$ to be a set of induced subgraphs of $G$ such that 
	for each $H \in \mathcal{H}(G, I, u)$,
	\begin{enumerate}[(H.1)]
		\item $u \in V(H)$, for every $v \in V(H) \cap S$, $L_G(v) \subseteq V(H)$, and if $v \neq u$, $L_G[v] \cap I = \emptyset$; and
		\item $H$ contains two $k$-paths $P$ and $Q$ such that $V(P) \cap V(Q)$ is either $\{u\}$ or $\{u, u^\prime\}$ for some $u^\prime \in L_G(u)$ and $((V(P) \cup V(Q)) \setminus \{u\}) \cap I = \emptyset$; and
		\item among all subgraphs of $G$ containing $u$ and satisfying (H.1) and (H.2), $\vert V(H) \cap S \vert$ is minimum.
	\end{enumerate} 
\end{definition}
Intuitively, if $\mathcal{H}(G, I, u) \neq \emptyset$, the token $t$ on $u \in I \cap S$ is $(G, I)$-movable if for every subgraph $H \in \mathcal{H}(G, I, u)$, at least one token from some vertex in $V(G - H) \cap I$ (i.e., ``outside $H$'') can be moved to some vertex in $V(H)$ (i.e., ``inside $H$'') via a $\mathsf{TS}$-sequence.
For example, it can be readily verified that in \figurename~\ref{fig:exa-HIu}(a), $\mathcal{H}(G, I, u) \neq \emptyset$ and in \figurename~\ref{fig:exa-HIu}(b), for each subgraph $H$ of $G$, either (H.1) or (H.2) does not hold, and therefore $\mathcal{H}(G, I, u) = \emptyset$, where $u = s_3$. 

\begin{figure}[!ht]
	\centering
	\begin{adjustbox}{max width=\textwidth}
	\begin{tikzpicture}[every node/.style={draw, thick, circle, minimum size=0.6cm, transform shape}, scale=0.8]
		\foreach \i in {0,1} {
			\foreach \j in {0} {
				\begin{scope}[shift={(8*\i, -3*\j)}]
					\foreach \x in {1,...,5} {
						\node [label=above:$s_{\x}$] (s\x) at (1.5*\x,0) {};
					}
					\draw[thick] (s1) -- (s2) -- (s3) -- (s4) -- (s5);
					\foreach \x/\y in {1/1, 2/2, 3/3.5, 4/4.5, 5/5.5, 6/7, 7/8} {
						\node (l\x) at (\y,-1) {};
					}
					\draw[thick] (s1) -- (l1) (s1) -- (l2);
					\draw[thick] (s3) -- (l3) (s3) -- (l4) (s3) -- (l5);
					\draw[thick] (s5) -- (l6) (s5) -- (l7);
					
					\ifthenelse{\i=0}{
						\node [fill=black, minimum size=0.5cm] at (s1.center) {};
						\node [fill=black, minimum size=0.5cm] at (s3.center) {};
						\node [fill=black, minimum size=0.5cm] at (s5.center) {};
						\node [fill=black, minimum size=0.5cm] at (l3.center) {};
						\node [fill=black, minimum size=0.5cm] at (l6.center) {};
						\node [fill=black, minimum size=0.5cm] at (l7.center) {};
						\begin{scope}[on background layer]
							\draw[very thick, gray] ([xshift=-0.25cm,yshift=0.5cm]s2.west) -- ([xshift=-0.25cm,yshift=-1.5cm]s2.west) -- node [draw=none, below, xshift=-1cm, yshift=0.2cm] {$H^l$} ([xshift=0.3cm,yshift=-1.5cm]s4.west) --  ([xshift=-0.15cm,yshift=0.5cm]s4.west) -- cycle;
							\draw[very thick] ([xshift=0.25cm,yshift=0.4cm]s4.east) -- ([xshift=0.25cm,yshift=-1.4cm]s4.east) node [draw=none, below, xshift=-0.2cm, yshift=0.2cm] {$H^r$} -- ([xshift=0.2cm,yshift=-1.4cm]s2.west) -- ([xshift=0.25cm,yshift=0.4cm]s2.east) -- cycle;
						\end{scope}
					}{
						\node [fill=black, minimum size=0.5cm] at (s1.center) {};
						\node [fill=black, minimum size=0.5cm] at (s3.center) {};
						\node [fill=black, minimum size=0.5cm] at (s5.center) {};
						\node [fill=black, minimum size=0.5cm] at (l3.center) {};
						\node [fill=black, minimum size=0.5cm] at (s4.center) {};
						\node [fill=black, minimum size=0.5cm] at (l4.center) {};
					}
				\end{scope}
			}
		}
		
		\node [draw=none] at (4.5,-1.75) {(a)};
		\node [draw=none] at (12.5,-1.75) {(b)};
	\end{tikzpicture}
	\end{adjustbox}
	\caption{Examples of different $3$-PVCs $I$ (marked with black tokens) of a caterpillar $G$ and the vertex $s_3 = u \in I$ where (a) $\mathcal{H}(G, I, u) = \{H^l, H^r\} \neq \emptyset$ and (b) $\mathcal{H}(G, I, u) = \emptyset$.}
	\label{fig:exa-HIu}
\end{figure}
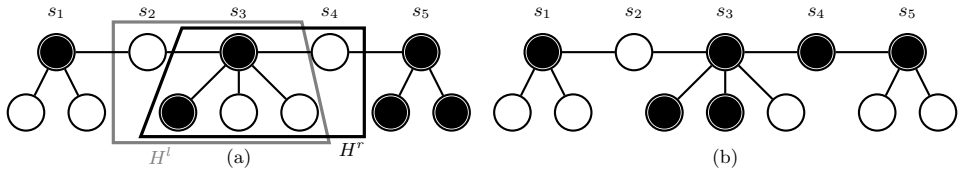

We now define some useful notations which will be used later in several statements.
\begin{definition}\label{def:PGIu}
	Let $I$ be a $k$-PVC ($k \geq 3$) of a caterpillar $G = (S \cup L, E)$, and let $u \in I \cap S$.
	We denote by $\mathcal{P}(G, I, u)$ the set of all $k$-paths in $G$ having either $u$ or one of its leaf-neighbors as an endpoint and none of their vertices other than $u$ is in $I$.
	Observe that $\mathcal{P}(G, I, u)$ can be partitioned into three subsets $\mathcal{P}_l(G, I, u)$, $\mathcal{P}_r(G, I, u)$, and $\mathcal{P}_c(G, I, u)$ where each member is a $k$-path containing $l(u)$, $r(u)$, and only $u$ and some of its leaf-neighbors, respectively.
	We denote by $\mathcal{A}(G, I, u)$ the set of all vertices $v \in I \setminus L_G(u)$ such that $\text{dist}_G(u, v) \leq k$ and $(V(P_{uv}) \setminus \{u, v\}) \cap I = \emptyset$.
\end{definition}

Indeed, the following lemma says that $\mathcal{H}(G, I, u)$ can be determined efficiently.
\begin{lemma}\label{lem:find-HIu}
	Let $I$ be a $k$-PVC ($k \geq 3$) of a $n$-vertex caterpillar $G = (S \cup L, E)$, and let $u \in I \cap S$.
	The set $\mathcal{H}(G, I, u)$ can be found in $O(n)$ time.
\end{lemma}
\begin{proof}
	First of all, if $\mathcal{H}(G, I, u) \neq \emptyset$, for each $H \in \mathcal{H}(G, I, u)$, we claim that $2k - 5 \leq \vert V(H) \cap S \vert \leq 2k - 1$.
	To see this, note that two $k$-paths satisfying (H.2) contain exactly one common vertex in $S$, namely $u$.
	Moreover, a $k$-path contains at least $k - 2$ and at most $k$ vertices in $S$.
	This minimum (resp., maximum) value can be obtained when two endpoints of the $k$-path are of degree $1$ (resp., more than $1$) in $G$.
	Then, $\vert V(H) \cap S \vert \geq 2(k-2) - 1 = 2k-5$ and $\vert V(H) \cap S \vert \leq 2k - 1$.
	
	From the definition, for every $H \in \mathcal{H}(G, I, u)$, $H$ contains $L_G[u] = \{u\} \cup L_G(u)$.
	It follows that to find $\mathcal{H}(G, I, u)$, it suffices to consider subgraphs of $G$ containing $L_G[u]$.
	Next, we show that given any subgraph $H$ of $G$ containing $L_G[u]$, (H.1) and (H.2) can be verified in $O(n)$ time.
	Indeed, it takes $O(\vert V(H) \cap S \vert) = O(n)$ time to verify (H.1).
	For verifying (H.2), we consider the set $\mathcal{P}(H, I, u)$.
	If $\mathcal{P}(H, I, u)$ is empty, clearly (H.2) is not satisfied.
	Otherwise, observe that for any two $k$-paths $P, Q$ in $H$ satisfying (H.2), $P$ and $Q$ cannot be both in either $\mathcal{P}_l(H, I, u)$ or $\mathcal{P}_r(H, I, u)$, otherwise they both respectively contain either $l(u)$ or $r(u)$, which contradicts (H.2).
	If both $\mathcal{P}_l(H, I, u)$ and $\mathcal{P}_r(H, I, u)$ are non-empty, (H.2) is satisfied: arbitrarily taking any $P$ from $\mathcal{P}_l(H, I, u)$ and $Q$ from $\mathcal{P}_r(H, I, u)$ would be sufficient.
	If both $\mathcal{P}_l(H, I, u)$ and $\mathcal{P}_r(H, I, u)$ are empty, any two $k$-paths $P, Q$ in $H$ satisfying (H.2), if they exist, must be both in $\mathcal{P}_c(H, I, u)$, and therefore have their endpoints in $L_H(u) = L_G(u)$.
	In this case, one can verify that (H.2) is satisfied if and only if $\vert L_G(u) \setminus I \vert \geq 3$.
	If exactly one of $\mathcal{P}_l(H, I, u)$ and $\mathcal{P}_r(H, I, u)$ is empty, say $\mathcal{P}_l(H, I, u)$, it follows that (H.2) is satisfied if and only if $\mathcal{P}_c(H, I, u)$ is non-empty: the only-if direction is trivial, and the if direction is proved by arbitrarily taking $P$ from $\mathcal{P}_l(H, I, u)$ and $Q$ from $\mathcal{P}_c(H, I, u)$.
	In each of the above cases, we have shown how to verify (H.2) in $O(1)$ time.
	Therefore, the running time of our verification for (H.2) depends on the time for constructing $\mathcal{P}(H, I, u)$, which can be done by checking $O((1 + \vert L_G(u) \vert)\max_{v \in V(G)}\deg_G(v)) = O(n)$ $k$-paths in $G$ which contains either $u$ or one of its leaf-neighbors as an endpoint.
	(Since $u$ is fixed, $\vert L_G(u) \vert$ is a constant.)
	
	Finally, we describe how to find $\mathcal{H}(G, I, u)$ in $O(n)$ time. 
	At the beginning, $\mathcal{H}(G, I, u) = \emptyset$.
	We initially start with the subgraph $H$ induced by vertices in $S$ of distance at most $k-3$ from $u$ and their leaf-neighbors.
	Note that $\vert V(H) \cap S \vert = 2k - 5$, and therefore none of $H$'s proper subgraphs is in $\mathcal{H}(G, I, u)$.
	As a result, if $H$ satisfies both (H.1) and (H.2), it certainly satisfies (H.3), and therefore we can stop and output $\mathcal{H}(G, I, u) = \{H\}$.
	Otherwise, we find the vertices $u_l$ and $u_r$ in $V(G - H) \cap S$ such that $r(u_l)$ and $l(u_r)$ are in $V(H)$, respectively.
	If none of $u_l$ and $u_r$ exists, we stop and output $\mathcal{H}(G, I, u) = \emptyset$.
	Otherwise, let $H^l$ (resp., $H^r$) be the graph obtained from $H$ by adding $u_l$ (resp., $u_r$), if it exists, and its corresponding leaf-neighbors and incident edges in $G$.
	If $u_l$ (resp., $u_r$) does not exist, we simply set $H^l = H$ (resp., $H^r = H$).
	Note that $\vert V(H^l) \cap S \vert = \vert V(H^r) \cap S \vert = 2k - 4$.
	As before, if either $H^l$ or $H^r$ satisfies (H.1) and (H.2), it immediately satisfies (H.3) because $H \notin \mathcal{H}(G, I, u)$, and therefore we can add it to $\mathcal{H}(G, I, u)$.
	If either $H^l$ or $H^r$ (or both) is in $\mathcal{H}(G, I, u)$, we stop and output $\mathcal{H}(G, I, u)$.
	Otherwise, we update $H$ by adding to it both $u_l$ and $u_r$ and their corresponding leaf-neighbors and incident edges, and repeat the above process.
	We stop the process when either $\mathcal{H}(G, I, u)$ is output or one of $\vert V(H) \cap S \vert$, $\vert V(H^l) \cap S \vert$, and $\vert V(H^r) \cap S \vert$ is larger than $2k - 1$.
	Our algorithm requires verifying (H.1) and (H.2) for at most seven subgraphs of $G$ containing $L_G[u]$ (e.g., see \figurename~\ref{fig:find-HIu} for $k=3$), and therefore runs in $O(n)$ time.
	In total, finding $\mathcal{H}(G, I, u)$ takes $O(n)$ time.
	
	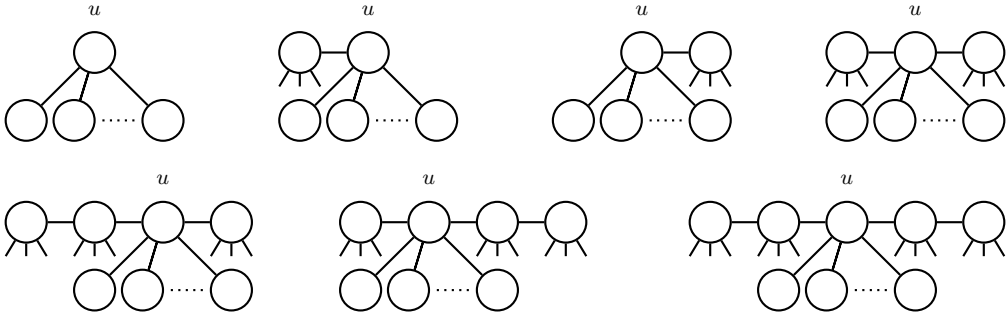
\begin{figure}[!ht]
		\centering
		\begin{adjustbox}{max width=\textwidth}
			\begin{tikzpicture}[every node/.style={draw, thick, circle, minimum size=0.6cm, transform shape}, scale=0.9]
				\foreach \i in {0,...,3} {
					\begin{scope}[shift={(4*\i, 0)}]
						\node [label=above:$u$] (s1) at (1, 0) {};
						\node (l11) at (0, -1) {};
						\node (l12) at (0.7, -1) {};
						\node (l13) at (2, -1) {};
						\draw [thick] (s1) -- (l11) (s1) -- (l12) -- (s1) -- (l13);
						\draw[thick, dotted] ([xshift={0.1cm}]l12.east) -- ([xshift={-0.1cm}]l13.west);
						
						\ifthenelse{\NOT \i=0 \AND \NOT \i=2}{
							\node (s2) at (0,0) {};
							\draw[thick] (s1) -- (s2);
							\coordinate (l21) at (-0.3,-0.5) {};
							\coordinate (l22) at (0,-0.5) {};
							\coordinate (l23) at (0.3,-0.5) {};
							\draw[thick] (s2) -- (l21) (s2) -- (l22) (s2) -- (l23);
						}{}
						\ifthenelse{\NOT \i=0 \AND \NOT \i=1}{
							\node (s3) at (2,0) {};
							\draw[thick] (s1) -- (s3);
							\coordinate (l31) at (1.7,-0.5) {};
							\coordinate (l32) at (2,-0.5) {};
							\coordinate (l33) at (2.3,-0.5) {};
							\draw[thick] (s3) -- (l31) (s3) -- (l32) (s3) -- (l33);
						}{}
					\end{scope}			
				}
				
				\foreach \i in {0,...,2} {
					\begin{scope}[shift={(5*\i+1, -2.5)}]
						\node [label=above:$u$] (s1) at (1, 0) {};
						\node (l11) at (0, -1) {};
						\node (l12) at (0.7, -1) {};
						\node (l13) at (2, -1) {};
						\draw [thick] (s1) -- (l11) (s1) -- (l12) -- (s1) -- (l13);
						\draw[thick, dotted] ([xshift={0.1cm}]l12.east) -- ([xshift={-0.1cm}]l13.west);
						
						\node (s2) at (0,0) {};
						\draw[thick] (s1) -- (s2);
						\coordinate (l21) at (-0.3,-0.5) {};
						\coordinate (l22) at (0,-0.5) {};
						\coordinate (l23) at (0.3,-0.5) {};
						\draw[thick] (s2) -- (l21) (s2) -- (l22) (s2) -- (l23);
						
						\node (s3) at (2,0) {};
						\draw[thick] (s1) -- (s3);
						\coordinate (l31) at (1.7,-0.5) {};
						\coordinate (l32) at (2,-0.5) {};
						\coordinate (l33) at (2.3,-0.5) {};
						\draw[thick] (s3) -- (l31) (s3) -- (l32) (s3) -- (l33);
						
						\ifthenelse{\NOT \i=1}{
							\node (s4) at (-1,0) {};
							\draw[thick] (s2) -- (s4);
							\coordinate (l41) at (-1.3,-0.5) {};
							\coordinate (l42) at (-1,-0.5) {};
							\coordinate (l43) at (-0.7,-0.5) {};
							\draw[thick] (s4) -- (l41) (s4) -- (l42) (s4) -- (l43);
						}{}
						\ifthenelse{\NOT \i=0}{
							\node (s5) at (3,0) {};
							\draw[thick] (s3) -- (s5);
							\coordinate (l51) at (2.7,-0.5) {};
							\coordinate (l52) at (3,-0.5) {};
							\coordinate (l53) at (3.3,-0.5) {};
							\draw[thick] (s5) -- (l51) (s5) -- (l52) (s5) -- (l53);
						}{}
					\end{scope}
					\ifthenelse{\i=0}{
						\hspace*{-1cm}
					}{
						\hspace*{1cm}
					}
				}
			\end{tikzpicture}
		\end{adjustbox}
		\caption{Seven possible candidates for verifying (H.1) and (H.2) when $k = 3$.}
		\label{fig:find-HIu}
	\end{figure}
\qed\end{proof}

Moreover, if $\mathcal{H}(G, I, u) \neq \emptyset$, we can further determine its size as follows.
\begin{lemma}\label{lem:HIu-size}
	Let $I$ be a $k$-PVC ($k \geq 3$) of a $n$-vertex caterpillar $G = (S \cup L, E)$, and let $u \in I \cap S$.
	Then, $0 \leq \vert \mathcal{H}(G, I, u) \vert \leq 2$.
	Moreover, suppose that $\mathcal{H}(G, I, u) \neq \emptyset$.
	Then, $\vert \mathcal{H}(G, I, u) \vert = 2$ if and only if $k = 3$, $\vert N_G(u) \cap S \vert = 2$, $\vert L_G(u) \setminus I \vert = 2$, and $(L_G[l(u)] \cup L_G[r(u)]) \cap I = \emptyset$.
\end{lemma}
\begin{proof}
	It follows from the proof of Lemma~\ref{lem:find-HIu} that $0 \leq \vert \mathcal{H}(G, I, u) \vert \leq 2$.
	Now, suppose that $\mathcal{H}(G, I, u) \neq \emptyset$.
	It remains to show that $\vert \mathcal{H}(G, I, u) \vert = 2$ if and only if $k = 3$, $\vert N_G(u) \cap S \vert = 2$, $\vert L_G(u) \setminus I \vert = 2$, and $(L_G[l(u)] \cup L_G[r(u)]) \cap I = \emptyset$.
	
	We now show the only-if direction.
	From the construction of $\mathcal{H}(G, I, u)$ in Lemma~\ref{lem:find-HIu}, we can assume w.l.o.g that $\mathcal{H}(G, I, u)$ has exactly two members $H^l$ and $H^r$, each of which satisfies (H.1)--(H.3).
	(For example, see \figurename~\ref{fig:exa-HIu}(a).)
	By definition, $H^l$ is not a subgraph of $H^r$ and vice versa.
	Moreover, $H^l$ (resp., $H^r$) is obtained from a subgraph $H$ of $G$ by adding the vertex $u_l$ (resp., $u_r$) and its leaf-neighbors, where $u_l \in V(G - H) \cap S$ (resp., $u_r \in V(G - H) \cap S$) is such that $r(u_l) \in V(H)$ (resp., $l(u_r) \in V(H)$).
	In particular, $V(H) = V(H^l) \cap V(H^r)$, $E(H) = E(H^l) \cap E(H^r)$, $L_G[u_l] = V(H^l) \setminus V(H^r)$, and $L_G[u_r] = V(H^r) \setminus V(H^l)$.
	Therefore, $H^l$ and $H^r$ respectively contain $l(u)$ and $r(u)$, which implies $\vert N_G(u) \cap S \vert = 2$.
	Let $P_{H^l}, Q_{H^l}$ (resp., $P_{H^r}, Q_{H^r}$) be two $k$-paths in $H^l$ (resp., $H^r$) satisfying (H.2).
	As before, note that $P_{H^l}$ and $Q_{H^l}$ ($P_{H^r}$ and $Q_{H^r}$) cannot be both in either $\mathcal{P}_l(G, I, u)$ or $\mathcal{P}_r(G, I, u)$.
	
	\begin{itemize}
		\item \textbf{Case~1: None of $P_{H^l}, Q_{H^l}, P_{H^r}, Q_{H^r}$ is in $\mathcal{P}_c(G, I, u)$.}
		We assume w.l.o.g that $P_{H^l}, P_{H^r} \in \mathcal{P}_l(G, I, u)$, $Q_{H^l}, Q_{H^r} \in \mathcal{P}_r(G, I, u)$.
		Note that $u_l \notin V(P_{H^r})\allowbreak \subseteq V(H^r)$ and $u_r \notin V(Q_{H^l}) \subseteq V(H^l)$.
		One can verify that $P_{H^r}$ and $Q_{H^l}$ are in $H$ and they both satisfy (H.2).
		Moreover, since $H^l$ and $H^r$ satisfy (H.1), so does $H$.
		As a result, $H$ satisfies both (H.1) and (H.2), and $\vert V(H) \cap S \vert < \vert V(H^l) \cap S \vert = \vert V(H^r) \cap S \vert$.
		This contradicts (H.3), therefore this case cannot happen.
		
		\item \textbf{Case~2: One of $P_{H^l}, Q_{H^l}, P_{H^r}, Q_{H^r}$ is in $\mathcal{P}_c(G, I, u)$.}
		W.l.o.g, assume that $P_{H^l} \in \mathcal{P}_c(G, I, u)$.
		By definition of $\mathcal{P}_c(G,I,u)$, the vertices of $P_{H^l}$ are in $L_G[u]$, which implies $k = 3$.
		Additionally, since $P_{H^l}$ satisfies (H.2), none of its endpoints is $u$ and therefore they are both not in $I$.
		Then, $\vert L_G(u) \setminus I \vert \geq 2$.
		Additionally, if $\vert L_G(u) \setminus I \vert \geq 3$, there must be a $3$-path $Q \neq P_{H^l}$ in $\mathcal{P}_c(G, I, u)$.
		One can verify that the graph $H^\star = G[L_G[u]]$ satisfies (H.1)--(H.3): $H^\star$ obviously satisfies (H.1) and (H.3); for (H.2), take the $3$-paths $P_{H^l}$ and $Q$.
		From the construction of $\mathcal{H}(G, I, u)$ in Lemma~\ref{lem:find-HIu}, it follows that $\mathcal{H}(G, I, u) = \{H^\star\}$, which contradicts our assumption that $\vert \mathcal{H}(G, I, u) \vert = 2$.
		Therefore, $\vert L_G(u) \setminus I \vert = 2$.
		Since $k = 3$ and $\vert L_G(u) \setminus I \vert = 2$, $H^\star$ does not satisfy (H.2), and therefore $H^\star \notin \mathcal{H}(G, I, u)$. 
		The construction of $\mathcal{H}(G, I, u)$ in Lemma~\ref{lem:find-HIu} implies that $H^l$ (resp., $H^r$) contains $l(u)$ (resp., $r(u)$) and their leaf-neighbors, and therefore by (H.1), we have $(L_G[l(u)] \cup L_G[r(u)]) \cap I = \emptyset$.
	\end{itemize}
	
	It remains to show the if direction.
	Since $k = 3$ and $\vert L_G(u) \setminus I \vert = 2$, there is a unique $3$-path $P = xuy$ in $\mathcal{P}_c(G, I, u)$, where $x, y \in L_G(u) \setminus I$.
	As a result, $H^\star = G[L_G[u]] \notin \mathcal{H}(G, I, u)$.
	Let us consider the graph $H^l$ (resp., $H^r$) obtained from $H^\star$ by adding $l(u)$ (resp., $r(u)$), its leaf-neighbors, and the corresponding incident edges.
	(Recall that we assumed $\vert N_G(u) \cap S \vert = 2$, which means $H^l$ and $H^r$ are distinct.)
	One can verify that $H^l$ satisfies (H.1)--(H.3): (H-1) and (H-3) are trivial; for (H.2), take the $3$-paths $P = xuy$ and $Q = xul(u)$. 
	(Recall that we assumed $L_G[l(u)] \cap I = \emptyset$.)
	Similarly, so does $H^r$.
	Therefore, the construction in Lemma~\ref{lem:find-HIu} will output $\mathcal{H}(G, I, u) = \{H^l, H^r\}$.
\qed\end{proof}

Let $I$ be a $k$-PVC ($k \geq 3$) of a caterpillar $G = (S \cup L, E)$.
We now characterize if a token on $u \in I \cap S$ is $(G, I)$-rigid.
The following simple observation will be used implicitly in several arguments:  if the token $t$ on $u \in I$ is $(G, I)$-rigid, then it is also $(G, J)$-rigid for every $k$-PVC $J$ that is reachable from $I$ via a $\mathsf{TS}$-sequence.

\begin{lemma}\label{lem:rigid-tokens-in-S}
	Let $I$ be a $k$-PVC ($k \geq 3$) of a $n$-vertex caterpillar $G = (S \cup L, E)$, and let $u \in I \cap S$. 
	Then, the token $t$ on $u$ is $(G, I)$-rigid if and only if one of the following conditions holds:
	\begin{enumerate}[(a)]
		\item $N_G(u) \subseteq I$ and every token $t_v$ on $v \in N_G(u)$ is $(G - u, I \cap V(G - u))$-rigid.
		\item $\mathcal{H}(G, I, u) \neq \emptyset$ and 
		either
		\begin{enumerate}[(b.1)]
			\item $G[L_G[u]] \in \mathcal{H}(G, I, u)$; or
			\item $\mathcal{A}(G, I, u) = \emptyset$; or
			\item for each $v \in \mathcal{A}(G, I, u)$, the token $t_v$ on $v$ is $(G - u, I \cap V(G - u))$-rigid; or
			\item there is at least one subgraph $H$ in $\mathcal{H}(G, I, u)$ satisfies that for every $(G - u, I \cap V(G - u))$-movable token $t_v$ on a vertex $v \in \mathcal{A}(G, I, u)$, there is no $\mathsf{TS}$-sequence in $G - u$ that slides either $t_v$ or one of the tokens placed on some vertices in $L_G(v)$ (if exist) to some vertex in $V(H)$.
		\end{enumerate}
	\end{enumerate}
\end{lemma}
\begin{proof}
	It is not hard to verify the if direction.
	(At first glance the case $\mathcal{H}(G, I, u)\allowbreak \neq \emptyset$ and $\mathcal{A}(G, I, u) = \emptyset$ may be a bit tricky. 
	In this case, we have $I \subseteq L_G[u]$, i.e., all tokens are either on $u$ or one of its leaf-neighbors, and because $\mathcal{H}(G, I, u) \neq \emptyset$ the token on $u$ is immediately $(G, I)$-rigid.)
	We prove the only-if direction by contraposition, i.e., we claim that if none of (a) and (b) hold, there is a $\mathsf{TS}$-sequence $\mathcal{S}$ in $G$ that slides $t$ from $u$ to some vertex in $N_G(u)$.
	More precisely, our new assumptions are as follows.
	\begin{enumerate}[(1)]
		\item Either (1.1) $N_G(u) \nsubseteq I$ or (1.2) $N_G(u) \subseteq I$ and there is $v \in N_G(u)$ such that the token $t_v$ on $v$ is $(G - u, I \cap V(G-u))$-movable; and
		\item Either (2.1) $\mathcal{H}(G, I, u) = \emptyset$ or (2.2) $\mathcal{H}(G, I, u) \neq \emptyset$ and none of (b.1)--(b.4) hold. 
	\end{enumerate}
	\begin{itemize}
		\item \textbf{Case~1: (1.1) and (2.1) hold.} 
		If $\mathcal{P}(G, I, u) = \emptyset$, it follows that for any $k$-path $P$ having either $u$ or one of its leaf-neighbors as an endpoint, there exists a vertex $v \in (V(P) \setminus \{u\}) \cap I$.
		If exactly one of $l(u)$ and $r(u)$ is not in $I$, we can immediately slide $t$ to it.
		Otherwise, since $N_G(u) \nsubseteq I$, there must be a vertex $u^\prime \in L_G(u) \setminus I$, and we can slide $t$ to $u^\prime$ immediately.
		
		Next, suppose that $\mathcal{P}(G, I, u) \neq \emptyset$.
		Let $\mathcal{H}$ be the set of seven candidates for being members of $\mathcal{H}(G, I, u)$ considered in the proof of Lemma~\ref{lem:find-HIu}.
		Note that $G[L_G[u]]$ always satisfies (H.1).
		Since $\mathcal{H}(G, I, u) = \emptyset$, for each member of $\mathcal{H}$, either (H.1) or (H.2) does not hold.
		First, we consider the following sub-case: none of the graphs in $\mathcal{H}$ satisfies (H.2).
		Then, it follows that $\mathcal{P}(G, I, u)$ must be either $\mathcal{P}_l(G, I, u)$, $\mathcal{P}_r(G, I, u)$, or $\mathcal{P}_c(G, I, u)$, otherwise, at least two of these sets must be non-empty and therefore (H.2) is satisfied.
		If $\mathcal{P}(G, I, u) = \mathcal{P}_c(G, I, u)$, we must have $k = 3$ and $\mathcal{P}(G, I, u)$ contains a single $3$-path $P$ whose endpoints are in $L_G(u) \setminus I$, otherwise $G[L_G[u]]$ satisfies (H.2), and therefore $G[L_G[u]] \in \mathcal{H}(G, I, u)$, which is a contradiction.
		Since both $G[L_G[u] \cup L_G[l(u)]]$ and $G[L_G[u] \cup L_G[r(u)]]$ do not satisfy (H.2), it follows that both $l(u)$ and $r(u)$ are in $I$, and therefore we can slide $t$ to one of $P$'s endpoints.
		If $\mathcal{P}(G, I, u) = \mathcal{P}_l(G, I, u)$, every $k$-path $Q$ not in $\mathcal{P}_l(G, I, u)$ satisfies $(V(Q) \setminus \{u\}) \cap I \neq \emptyset$.
		Additionally, from the definition of $\mathcal{P}(G, I, u)$, we must have $l(u) \notin I$.
		Then, we can move $t$ from $u$ to $l(u)$ immediately.
		One can argue similarly for the case $\mathcal{P}(G, I, u) = \mathcal{P}_r(G, I, u)$.
		It remains to consider the second sub-case: some member of $\mathcal{H}$ satisfies (H.2). 
		For each $H \in \mathcal{H}$ satisfying (H.2), it cannot satisfy (H.1).
		Since $H$ does not satisfy (H.1), there exists a vertex $v_H \in V(H) \cap S$ such that $v_H \neq u$ and $L_G[v_H] \cap I \neq \emptyset$.
		(Recall that for every $H \in \mathcal{H}$, we always have $L_G[u] \subseteq V(H)$.)
		For each such $v_H$, if it is not in $I$, we immediately slide a token from some vertex in $L_G(v_H) \cap I$ to $v_H$, otherwise we do nothing.
		Since any $k$-path covered by a leaf is also covered by its unique neighbor in $G$, these token-slides always result new $k$-PVCs.
		In this way, via a $\mathsf{TS}$-sequence, we finally obtain a new $k$-PVC $J$.
		From the construction of $J$, for every $H \in \mathcal{H}$, any $v_H \in (V(H) \cap S) \setminus \{u\}$ satisfying $L_G[v_H] \cap J \neq \emptyset$ must be in $J$.
		Observe that with the new $k$-PVC $J$, no $H \in \mathcal{H}$ satisfies (H.2).
		Note that for each $H \in \mathcal{H}$ not satisfying (H.2) with $I$, it also does not satisfy (H.2) with $J$.
		Now, suppose to the contrary that there exist two $k$-paths $P_H$ and $Q_H$ of a graph $H \in \mathcal{H}$ satisfying (H.2) with $J$.
		Since $H$ satisfies (H.2), it does not satisfy (H.1), which implies that there exists a vertex $w \in V(H) \cap S$ such that $w \neq u$ and $L_G[w] \cap J \neq \emptyset$.
		From the construction of $J$, we have $w \in J$.
		From the definition of $H$, we have $w \in (V(P_H) \cup V(Q_H)) \setminus \{u\}$, which contradicts the assumption that $P_H, Q_H$ satisfy (H.2).
		As a result, we are back to the first sub-case.  
		
		\item \textbf{Case~2: (1.2) and (2.1) hold.}
		Since $t_v$ is $(G-u, I \cap V(G-u))$-movable, there is a $\mathsf{TS}$-sequence $\mathcal{S}^\prime = \langle I_0, \dots, I_p \rangle$ in $G - u$ that slides $t_v$ to one of its neighbors $w \in N_G(v) \setminus \{u\}$.
		Moreover, for any $k$-PVC $I^\prime$ of $G - u$, the set $I^\prime \cup \{u\}$ forms a $k$-PVC of $G$.
		It follows that $\mathcal{S} = \langle I_0 \cup \{u\}, \dots, I_p \cup \{u\} \rangle$ is a $\mathsf{TS}$-sequence in $G$ that slides $t_v$ to $w$.
		Moreover, after moving $t_v$, the only neighbor of $u$ having no token is $v$, which means once $t_v$ is placed on $w$, we can immediately slide $t$ from $u$ to $v$.
		
		\item \textbf{Case~3: (1.1) and (2.2) hold.}
		Since (b.2)--(b.4) do not hold, it follows that for each $H \in \mathcal{H}(G, I, u)$, there exists a $(G - u, I \cap V(G - u))$-movable token $t_v$ on some vertex $v \in \mathcal{A}(G, I, u)$ such that some $\mathsf{TS}$-sequence $\mathcal{S}^\prime$ in $G - u$ slides either $t_v$ or one of the tokens placed on some vertices in $L_G(v)$ (if exist) to some vertex $w \in V(H)$.
		For any $k$-PVC $I^\prime$ of $G - u$, $I^\prime \cup \{u\}$ forms a $k$-PVC of $G$.
		Thus, by adding $u$ to each member of $\mathcal{S}^\prime$, we obtain a $\mathsf{TS}$-sequence in $G$ that slides either $t_v$ or one of the tokens placed on some vertices in $L_G(v)$ (if exist) to $w$.
		By definition, $\mathcal{A}(G, I, u) \cap L_G(u) \neq \emptyset$.
		Let $P_H$ and $Q_H$ be two $k$-paths in $H$ satisfying (H.2).
		Since (b.1) does not hold, it follows that at least one of $P_H$ and $Q_H$, say $P_H$, must be in either $\mathcal{P}_l(G, I, u)$ or $\mathcal{P}_r(G, I, u)$.
		As a result, we can assume w.l.o.g that $w \in V(P_H) \subseteq V(H)$.
		(Otherwise, $\vert V(H) \cap S \vert$ is not minimum, which contradicts $H \in \mathcal{H}(G, I, u)$.)
		Then, once a token is placed on $w$, we can immediately move $t$ from $u$ to one of its neighbors in $V(Q_H)$.
		
		\item \textbf{Case~4: (1.2) and (2.2) hold.} Observe that the condition $N_G(u) \subseteq I$ implies $\mathcal{H}(G, I, u) = \emptyset$: for any subgraph $H$ of $G$ containing $u$, (H.1) does not hold. 
		Therefore, this case cannot happen.
	\end{itemize}
\qed\end{proof}

Observe that if $k \geq 4$, $G[L_G[u]] \notin \mathcal{H}(G, I, u)$.
Additionally, Lemma~\ref{lem:HIu-size} implies that when $k \geq 4$, $\mathcal{H}(G, I, u)$ has at most one member.
The following lemma says that Lemma~\ref{lem:rigid-tokens-in-S}(b.4) can be verified efficiently when $k \geq 4$.
\begin{lemma}\label{lem:check-b4}
	Let $I$ be a $k$-PVC ($k \geq 4$) of a $n$-vertex caterpillar $G = (S \cup L, E)$, and let $u \in I \cap S$.
	Suppose that $\mathcal{H}(G, I, u) \neq \emptyset$, and there exists a $(G - u, I \cap V(G - u))$-movable token $t_v$ on some $v \in \mathcal{A}(G, I, u)$.
	Let $H$ be the unique member in $\mathcal{H}(G, I, u)$.
	Then, for each $v \in \mathcal{A}(G, I, u)$ such that the token $t_v$ on $v$ is $(G - u, I \cap V(G - u))$-movable, one can decide in $O(n)$ time if there is a $\mathsf{TS}$-sequence in $G - u$ that moves either $t_v$ or one of the tokens on some vertex in $L_G(v)$ (if exist) to some vertex in $V(H)$.
\end{lemma}

\begin{proof}
	By definition, $(P_{uv} \setminus \{u, v\}) \cap I = \emptyset$.
	Thus, if $v \in L$, $t_v$ can be immediately slid to $v$'s unique neighbor in $S$ (which is also in $V(P_{uv}) \setminus \{u, v\}$).
	Then, we can assume w.l.o.g that for every $v \in \mathcal{A}(G, I, u)$ such that the token $t_v$ on $v$ is $(G - u, I \cap V(G - u))$-movable, $v$ is also in $S$.
	
	Note that $k-2 \leq \text{dist}_G(u, v) \leq k$.
	Moreover, if $\text{dist}_G(u, v) = k$, the token $t_v$ can only be slid to its unique neighbor in $P_{uv}$; otherwise, since no token can ``jump'' to some vertex in $V(P_{uv}) \setminus \{u, v\}$ or one of its leaf-neighbors before $t_v$ moves, some uncovered $k$-path exists.
	Once $t_v$ is moved, we are indeed considering the case $\text{dist}_G(u, v) = k-1$. 
	As a result, we can assume w.l.o.g that $k-2 \leq \text{dist}_G(u, v) \leq k-1$.
	
	Suppose that $G[S] = s_1\dots s_\ell$.
	If $v$ and $v^\prime$ are not in the same component of $G - u$ for all $v^\prime \in V(H) \setminus L_G[u]$, clearly the answer is ``no''.
	Otherwise, let $G_v$ be the component of $G - u$ containing $v$ and some $v^\prime \in V(H) \setminus L_G[u]$.
	Since $G_v$ contains at least two distinct vertices $v, v^\prime$, it also contains exactly one vertex $r \in \{s_1, s_\ell\}$.
	Let $H_v = G_v - (V(H) \cap V(G_v))$.
	Note that $H_v$ contains both $v$ and $r$.
	Let $P(H_v, k, r) = \{T_1(r), \dots, T_p(r)\}$ and $I(H_v, k, r) = \{v_1, \dots, v_p\}$ be respectively the partition and the minimum $k$-PVC of $H_v$ obtained in $O(|V(H_v)|)$ time by running $\mathtt{Partition}(H_v, k, r)$ (Algorithm~\ref{algo:partition}), where $p = \psi_k(H_v)$---the size of a minimum $k$-PVC of $H_v$.
	In particular, $v_i \in V(T_i(r)) \cap S$ ($i \in \{1, \dots, p\}$), $v \in V(T_1(r))$, and $r \in V(T_p(r)) \cap S$.
	
	Additionally, we now characterize the position of $t_v$ in $T_1(r)$.
	
	\begin{claim}\label{clm:v-neq-v1}
		We always have $v \neq v_1$.
		Additionally, if there is a $\mathsf{TS}$-sequence $\mathcal{S}$ in $G_v$ that slides $t_v$ to some $v^\prime \in L_G(v)$, $\mathcal{S}$ must slide a token from some vertex in $H_v - T_1(r)$ to $v_1$ before sliding $t_v$ from $v$ to $v^\prime$. 
	\end{claim}
	\begin{proof}
		We first show that $v \neq v_1$.
		Recall that we assumed $k-2 \leq \text{dist}_G(u, v) \leq k-1$.
		Suppose to the contrary that $v = v_1$.
		Then, from the definition of $H$ and $v$, for every $x \in V(P_{uv_1}) \setminus \{u, v_1\}$, we have $L_G[x] \cap I = \emptyset$.
		(Recall that we assumed w.l.o.g that $v \in S$, that is, if $v \in L$, we moved $t_v$ to $v$'s unique neighbor in $S$ and regarded that vertex as $v$.)
		Let $a \in V(H)$ and $b \in V(H_v)$ be such that $ab \in E(G_v)$.
		Then, $a, b \in V(P_{uv_1})$.
		Note that by definition of $T_1(r)$, Algorithm~\ref{algo:partition} starts from either $b$ (if it is a leaf in $H_v$) or one of its leaf-neighbors and ``moves toward $r$'' along the spine $S$ until finding $v_1$ such that $T_1(r)$ contains a $k$-path while $T_1(r) - v_1$ does not, which implies that $\text{dist}_G(b, v_1) \geq k - 3$.
		Moreover, $\text{dist}_G(u, v_1) = \text{dist}_G(u, a) + \text{dist}_G(a, b) + \text{dist}_G(b, v_1) \geq (k-3) + 1 + (k-3) = 2k - 5$.
		Therefore, $2k-5 \leq \text{dist}_G(u, v_1) = \text{dist}_G(u, v) \leq k-1$, which means $k \leq 4$.
		Additionally, $k \geq 4$.
		Finally, we obtain $k = 4$.
		It follows that $\text{dist}_G(u, v_1) = 3$ and therefore $P_{uv_1} = uabv_1$. 
		By definition of $H$, every leaf of $a$ must not be in $I$.
		By definition of $T_1(r)$ and the assumption $v = v_1$, every leaf of $b$ must not be in $I$, otherwise a leaf of $b$ in $I$ also satisfies the definition of $v$, which contradicts our assumption that any such vertex must be in $S$.
		However, $a$, $b$ and their ``empty'' leaf-neighbors form at least one uncovered $P_4$, which contradicts the fact that $I$ is a $4$-PVC.
		Therefore, $v \neq v_1$.
		
		Let $\mathcal{S}$ be a $\mathsf{TS}$-sequence in $G_v$ that slides $t_v$ to some $v^\prime \in L_G(v)$.
		Suppose to the contrary that $S$ does not slide any token to $v_1$.
		We show that immediately sliding $t_v$ from $v$ to $v^\prime$ results some uncovered $k$-path.
		As before, we have $\text{dist}_G(u, v_1) \geq 2k - 5$.
		Obviously, if $k \geq 5$, we have $\text{dist}_G(u, v_1) \geq 2k - 5 \geq k$, and sliding $t_v$ from $v$ to $v^\prime$ results an uncovered $k$-path $P_{v_1u^\prime}$, where $u^\prime$ is the unique neighbor of $u$ in $P_{uv_1}$.
		Now, if $k = 4$, one can readily verify that in both cases $\text{dist}_G(u, v_1) = 3$ and $\text{dist}_G(u, v_1) \geq 4$, sliding $t_v$ from $v$ to $v^\prime$ results in some uncovered $4$-path.
	\qed\end{proof} 
	
	Observe that if $\vert V(T_i(r)) \cap I \vert = 1$ for every $i \in \{1, \dots, p\}$, the set $I \cap V(H_v)$ is indeed a minimum $k$-PVC of $H_v$, and therefore we answer ``no'' in this case.
	Thus, it remains to consider the case when there exists some index $i \in \{1, \dots, p\}$ such that $\vert V(T_i(r)) \cap I \vert \geq 2$.
	For convenience, we also use $i$ to denote the smallest member among such indices. 
	In this case, we claim that the answer is always ``yes'', that is, there exists a $\mathsf{TS}$-sequence in $G_v$ that slides either $t_v$ or one of the tokens on some vertex in $L_G(v)$ (if exist) to some vertex in $V(H)$.
	(See Claims~\ref{clm:i-equal-1} and~\ref{clm:i-gt-1}.)
	
	Claim~\ref{clm:v-neq-v1} implies that one can always slide $t_v$ from $v$ to some vertex in $V(H)$ if there exists a $\mathsf{TS}$-sequence in $G_v$ that slides $t_v$ to some $v^\prime \in L_G(v)$.
	(Slide a token to $v_1$ first, and then slide $t_v$.)
	As a result, we can now assume w.l.o.g that for every $(G-u, I \cap V(G-u))$-movable token $t_v$, any $\mathsf{TS}$-sequence in $G_v$ that slides $t_v$, if exists, must move it to either $l(v)$ or $r(v)$.
	
	To simplify our proof, we further assume that $r = s_1$, that is, we only care about what happens in the ``left-hand side'' of $H$. 
	The case $r = s_\ell$ can be argued similarly.
	Let $w$ be the (unique) vertex in $H$ where $l(w) \in V(T_1(s_1))$.
	We first consider the case $i = 1$.
	
	\begin{claim}\label{clm:i-equal-1}
		When $i = 1$, one can slide either $t_v$ or one of the tokens placed on some vertices in $L_G(v)$ (if exist) to $w$.
	\end{claim}
	\begin{proof}
		If $v_1 \in I$, we can immediately slide $t_v$, because $v_1$ covers all $k$-paths in $T_1(s_1)$ and $v \neq v_1$.
		Thus, suppose that $v_1 \notin I$.
		As before, from the definition of $H$ and $v$, for every $x \in V(P_{uv}) \setminus \{u, v\}$, we have $L_G[x] \cap I = \emptyset$.
		Now, if there exists $a \in V(T_1(s_1)) \cap I$ such that $(V(P_{av_1}) \setminus \{a, v_1\}) \cap I = \emptyset$ and $a \neq v$, since $T_1(s_1) - v_1$ does not contain any $k$-path, we can directly slide the token on $a$ to $v_1$.
		Once a token is placed on $v_1$, as before, we can slide $t_v$ to $w$.
		If such vertex $a$ does not exist, we must have $V(T_1(s_1)) \cap I = L_G[v] \cap I$.
		As $t_v$ is $(G - u, I \cap V(G - u))$-movable, there is a $\mathsf{TS}$-sequence $\mathcal{S}$ in $G_v$ (the component of $G - u$ containing $v$) that slides $t_v$ to some vertex $v^\prime \in N_G(v)$.
		From our assumption, $v^\prime \in \{l(v), r(v)\}$.
		If $v^\prime = r(v)$, we can then slide a token on some vertex in $L_G(v)$ (such a token exists because $i = 1$) to $v_1$, and then slide $t_v$, which is now on $r(v)$, to $w$.
		If $v^\prime = l(v)$, we can then slide $t_v$, which is now on $l(v)$, to $v_1$, and then slide a token on some vertex in $L_G(v)$ to $w$.
		Our proof is complete.
	\qed\end{proof}
	
	It remains to consider the case $i > 1$.
	More precisely, we show that
	\begin{claim}\label{clm:i-gt-1}
		When $1 < i \leq p$, one can slide a token in $V(T_i(s_1)) \cap I$ to $v_{i-1} \in V(T_{i-1}(s_1))$.
	\end{claim}
	\begin{proof}
		Recall that we assumed that any $\mathsf{TS}$-sequence in $G_v$ that slides $t_v$, if exists, must move it to either $l(v)$ or $r(v)$.
		Indeed, if there is a $\mathsf{TS}$-sequence in $G_v$ that slides $t_v$ to $r(v)$, we can perform such sequence to obtain a new $k$-PVC where we are certain that there exists a $\mathsf{TS}$-sequence in $G_v$ that slides $t_v$ to $l(v)$.
		Moreover, if there is a $\mathsf{TS}$-sequence in $G_v$ that slides $t_v$ to $l(v)$, we can perform such a sequence except the final step of sliding $t_v$ to $l(v)$ to obtain a new $k$-PVC where we are certain that one can immediately slide $t_v$ to $l(v)$.
		In other words, we can further assume w.l.o.g that one can always slide $t_v$ to $l(v)$ immediately.
		
		Recall that $i$ is the smallest index such that $T_i(s_1)$ contains at least two tokens.
		As a result, for each index $j$ such that $1 \leq j < i$, there is exactly one token placed on some vertex in $T_j(s_1)$, say $v^j$.
		In particular, $v^1 = v$.
		Additionally, let $v^i \in V(T_i(s_1)) \cap I$ be such that $(V(P_{v^i v_{i-1}}) \setminus \{v^i, v_{i-1}\}) \cap I = \emptyset$.
		Note that if $v^i \in L$, there must be no token placed on its unique neighbor in $S$; otherwise, it contradicts the definition of $v^i$.
		We can thus assume w.l.o.g that $v^j \in S$ for $1 \leq j \leq i$; otherwise, we can immediately slide the token on $v^j$ to its unique neighbor in $S$ and continue working with the resulting $k$-PVC.
		
		\begin{itemize}
			\item \textbf{When $1 < i < p$.}
			We first consider the case $v_i \in I$.
			If $v^i \neq v_i$ then we can slide the token on $v^i$ to the unique neighbor $x$ of $v_{i-1}$ in $T_i(s_1)$.
			This can be done because $v_i$ covers all $k$-paths in $T_i(s_1)$ and $(V(P_{v^i v_{i-1}}) \setminus \{v^i, v_{i-1}\}) \cap I = \emptyset$.
			At this point, if $v_{i-1} \notin I$, we can immediately slide the token on $x$ to $v_{i-1}$.
			If $v_{i-1} \in I$, we simply slide the token on $v_{i-1}$ to its right-neighbor $r(v_{i-1})$ (recall that $\vert V(T_{i-1}(s_1)) \cap I \vert = 1$) and then slide the token on $x$ to $v_{i-1}$.
			These $\mathsf{TS}$-moves do not result in any uncovered $k$-path because $k \geq 4$, $v_i \in I$, and a token has already been placed on $x$.
			On the other hand, if $v^i = v_i$, we must have $V(T_i(s_1)) \cap I = L_G[v_i] \cap I$, otherwise, it contradicts the definition of $v^i$.
			We now show that one can slide the token on $v_i$ to either $l(v_i)$ or $r(v_i)$.
			Once the token on $v_i$ is moved, we can then move one or more tokens on leaf-neighbors of $v_i$ to some vertex in $T_i(s_1)$ other than $v_i$ or its leaf-neighbors, and we are back to the case $v^i \neq v_i$.
			Observe that $r(v_i) \notin I$.
			If $l(v_i) \in I$, we can immediately move the token on $v_i$ to $r(v_i)$.
			Thus, it remains to consider the case $l(v_i) \notin I$.
			From the definition of $v^j$ ($1 < j \leq i$) and the observation that $V(T_i(s_1)) \cap I = L_G[v_i] \cap I$, for every $x \in V(P_{v^j v^{j-1}}) \setminus \{v^j, v^{j-1}\}$, we have $L_G[x] \cap I = \emptyset$.
			Additionally, we have $\text{dist}_G(l(v), v^2) = \text{dist}_G(v, v^2) - 1 = \text{dist}_G(v^1, v^2) - 1 \leq k - 1$.
			Recall that we assumed $t_v$ can be slid to $l(v)$ immediately.
			Then, it follows that the token on $v^2$ can indeed be slid to $l(v^2)$ immediately.
			By repeatedly applying this argument, we finally obtain a $\mathsf{TS}$-sequence that moves the token on $v^j$ to $l(v^j)$, for $1 \leq j \leq i$, that is, such a sequence finally moves the token on $v_i = v^i$ to $l(v_i)$, as we required.
			
			It remains to consider the case $v_i \notin I$.
			Let $x \in V(T_i(s_1)) \cap I$ be such that $(V(P_{xv_i}) \setminus \{x, v_i\}) \cap I = \emptyset$.
			If there is such a vertex $x$ with $x \neq v^i$, since $T_i(s_1) - v_i$ does not contain any $k$-path, we can simply slide the token on $x$ to $v_i$ along $P_{xv_i}$, and we are then back to the case $v_i \in I$.
			Otherwise, it follows that all tokens in $V(T_i(s_1)) \cap I$ must be placed either on $x = v^i$ or one of its leaf-neighbors.
			Observe that $l(x) \notin I$.
			Since $t_v$ can only be slid to $l(v)$, as before, we can indeed show that there is a $\mathsf{TS}$-sequence that moves the token on $v^j$ to $l(v^j)$, for $1 \leq j \leq i$, that is, such a sequence finally moves the token on $x = v^i$ to $l(x)$.
			At this point, we can slide a token from some leaf-neighbor of $x$ to $x$ itself, slide the token on $l(x)$ to $v_i$ along $P_{l(x)v_i}$, and then we are back to the case $v_i \in I$.
			
			\item \textbf{When $i = p$.} 
			If $v_p \in I$ and $v^p \neq v_p$, as before, one can construct a $\mathsf{TS}$-sequence that slides the token on $v^p$ to $l(v_{p-1})$, then slides the token on $v_{p-1}$ away, if it exists, and finally slide the token on $l(v_{p-1})$ to $v_{p-1}$.
			If $v_p \in I$ and $v^p = v_p$, let $x \in V(T_p(s_1)) \cap I$ be such that $(V(P_{xv_p}) \setminus \{x, v_p\}) \cap I = \emptyset$.
			Note that if $v_p = s_j$ then $x$ must be placed on some vertex in $\bigcup_{j^\prime=1}^j L_G[s_{j^\prime}] \setminus \{s_j\}$.
			If $x \in L_G(v_p)$, we can simply slide the token on $v_p$ to $l(v_p)$ as in the previous case (which can be done by immediately sliding the token on $v$ to $l(v)$ and so on) and then immediately slide the token on $x$ to $v_p$ and then to $r(v_p)$, and finally immediately slide the token on $l(v_p)$ back to $v_p$. 
			(This can be done because $k \geq 4$.)
			The current token on $r(v_p)$ can now be slid to $v_{p-1}$ as before.
			Now, if $x \notin L_G(v_p)$, since $k \geq 4$ and $T_p(s_1) - v_p$ contains no $k$-path, one can slide the token $t_x$ on $x$ to $l(v_p)$, then slide the token on $v_p$ to $r(v_p)$, and finally slide $t_x$, which is currently placed on $l(v_p)$, to $v_p$.
			Again, the token on $r(v_p)$ can now be slid to $v_{p-1}$.
		\end{itemize}
	\qed\end{proof}
	
	By repeatedly applying Claim~\ref{clm:i-gt-1} while $i > 1$ and finally applying Claim~\ref{clm:i-equal-1} when $i = 1$, one can indeed construct a $\mathsf{TS}$-sequence in $G_v$ that slides either $t_v$ or one of the tokens on some vertex in $L_G(v)$ (if exist) to some vertex in $V(H)$.
	We remind that only the existence of such a $\mathsf{TS}$-sequence is required in deciding whether $t_v$ or a token in $L_G(v)$ (if exists) can be moved to $V(H)$.
	As a result, one can indeed verify that deciding whether $t_v$ or a token in $L_G(v)$ (if exists) can be moved to $V(H)$ can indeed be done in $O(n)$ time.
	Our proof is complete.
\qed\end{proof}

\begin{remark}\label{rem:k3-is-harder}
	We remark that deciding whether Lemma~\ref{lem:rigid-tokens-in-S}(b.4) can be verified in polynomial time when $k = 3$ may require further insightful ideas. 
	Unlike the case $k \geq 4$, even when the set $I \cap V(H_v)$ is not minimum, it is possible that no tokens in $V(H_v)$ can be moved to some vertex in $V(H)$. 
	\figurename~\ref{fig:k3-is-harder} illustrates an example of this situation.
	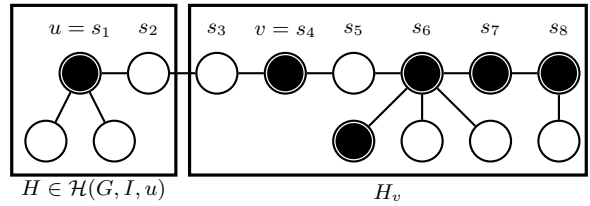
\begin{figure}[!ht]
		\centering
		\begin{adjustbox}{max width=\textwidth}
		\begin{tikzpicture}[every node/.style={draw, thick, circle, minimum size=0.6cm, transform shape}, scale=0.9]
			\begin{scope}
				\foreach \i in {1,...,8} {
					\ifthenelse{\i=1}{
						\node [label={[label distance=-0.3cm]above:$u = s_{\i}$}] (s\i) at (\i, 0) {};
					}{
						\ifthenelse{\i=4}{
							\node [label={[label distance=-0.3cm]above:$v = s_{\i}$}] (s\i) at (\i, 0) {};
						}{
							\node [label=above:{$s_{\i}$}] (s\i) at (\i, 0) {};
						}
					}
				}
				\draw [thick] (s1) -- (s2) -- (s3) -- (s4) -- (s5) -- (s6) -- (s7) -- (s8);
				
				\node (l11) at (0.5,-1) {};
				\node (l12) at (1.5,-1) {};
				\draw [thick] (s1) -- (l11) (s1) -- (l12);
				
				\node (l61) at (5,-1) {};
				\node (l62) at (6,-1) {};
				\node (l63) at (7,-1) {};
				\draw [thick] (s6) -- (l61) (s6) -- (l62) (s6) -- (l63);
				
				\node (l81) at (8,-1) {};
				\draw [thick] (s8) -- (l81);
				
				\draw [very thick] (0,1) -- (0,-1.5) -- node [below, draw=none, yshift={1cm}] {$H \in \mathcal{H}(G, I, u)$} (2.4, -1.5) -- (2.4, 1) -- cycle;
				\draw [very thick] (2.6,1) -- (2.6,-1.5) -- node [below, draw=none, yshift={0.15cm}] {$H_v$} (8.4, -1.5) -- (8.4, 1) -- cycle;
				
				\node [fill=black, minimum size=0.5cm] at (s1.center) {};
				\node [fill=black, minimum size=0.5cm] at (s4.center) {};
				\node [fill=black, minimum size=0.5cm] at (s6.center) {};
				\node [fill=black, minimum size=0.5cm] at (l61.center) {};
				\node [fill=black, minimum size=0.5cm] at (s7.center) {};
				\node [fill=black, minimum size=0.5cm] at (s8.center) {};
			\end{scope}
		\end{tikzpicture}
		\end{adjustbox}
		\caption{Illustration for Remark~\ref{rem:k3-is-harder}. Here $I$ is a $3$-PVC whose tokens are marked with black colors, $u = s_1$, and $v = s_4$. The concepts and notations follow Lemma~\ref{lem:check-b4}. The token on $v$ can be slid immediately to $s_3$.}
		\label{fig:k3-is-harder}
	\end{figure}
	In \figurename~\ref{fig:k3-is-harder}, intuitively, it seems that the reason no token can be slid to some vertex in $V(H)$ is because the token on $s_6$ somehow ``blocks'' the movement of all other tokens on its ``right-hand side''.
	Characterizing such a token on $s_6$ in general may require verifying whether some tokens can be slid in $H_v - L_G[s_6]$ to both $s_6$'s left-neighbor $s_5$ and right-neighbor $s_7$ in polynomial time.
	More interestingly, $s_6$ may not be the only vertex whose token having such ``blocking property''.
\end{remark}

To conclude this section, we show that
\begin{lemma}\label{lem:find-rigid-tokens}
	Let $I$ be a $k$-PVC ($k \geq 4$) of a caterpillar $G = (S \cup L, E)$.
	One can find the set $\mathcal{R}(G, I)$ of all $(G, I)$-rigid tokens in $O(n^3)$ time, where $n = \vert V(G) \vert$.
\end{lemma}
\begin{proof}
	Using Lemmas~\ref{lem:rigid-tokens-in-L} and~\ref{lem:rigid-tokens-in-S}, one can naturally design a recursive algorithm to decide whether a token $t$ on some vertex $u \in I$ is $(G, I)$-rigid, where $I$ is a $k$-PVC ($k \geq 4$) of a caterpillar $G = (S \cup L, E)$.
	Such an algorithm should be called at most $O(n)$ time (at most once for each token), and each call requires $O(n)$ time for (possibly) checking if $\mathcal{H}(G, I, u)$ exists (Lemma~\ref{lem:find-HIu}) and $O(n)$ time for (possibly) checking if a $k$-PVC is minimum (Algorithm~\ref{algo:partition} gives the minimum size in $O(n)$ time~\cite{HoangSY22}). 
	Additionally, Lemma~\ref{lem:check-b4} makes sure that the condition (b.4) of Lemma~\ref{lem:rigid-tokens-in-S} can indeed be verified in $O(n)$ time when $k \geq 4$.
	In total, it takes $O(n^2)$ time to verify if the token on a vertex $u \in I$ is $(G, I)$-rigid, and therefore it takes $O(n^3)$ time to find $\mathcal{R}(G, I)$.
\qed\end{proof}

\section{Our Algorithm}
\label{sec:algorithm}

In this section, we use $(G, I, J)$ to denote an instance of \textsc{$k$-PVCR} under $\mathsf{TS}$ whose input contains a caterpillar $G = (S \cup L, E)$ and two $k$-PVCs $I, J$ ($k \geq 4$) of $G$.
Algorithm~\ref{algo:PVCR-k4} decides whether there is a $\mathsf{TS}$-sequence between $I$ and $J$ in polynomial time.
The running time of Algorithm~\ref{algo:PVCR-k4} depends on the time of finding $\mathcal{R}(G, I)$ and $\mathcal{R}(G, J)$, which, as we showed before, is $O(n^3)$.
The rest of this section is devoted to proving its correctness.
More precisely, Lemma~\ref{lem:no-instances} implies that Lines~\ref{algo-line:PVCR-k4-s1}--\ref{algo-line:PVCR-k4-e1} are correct.
Lemma~\ref{lem:remove-rigid-tokens} explains the correctness of Line~\ref{algo-line:PVCR-k4-se2}.
Lemma~\ref{lem:no-rigid-tokens} implies the correctness of Lines~\ref{algo-line:PVCR-k4-s3}--\ref{algo-line:PVCR-k4-e3}.

\begin{algorithm}[!ht]
	\KwIn{A caterpillar $G = (S \cup L, E)$ on $n$ vertices and two $k$-PVCs ($k \geq 4$) $I, J$ of $G$.}
	\KwOut{\textsc{Yes} if there is a $\mathsf{TS}$-sequence between $I$ and $J$, and \textsc{No} otherwise.}
	\SetArgSty{textbb}  
	
	\If{$\vert I \vert \neq \vert J \vert$\label{algo-line:PVCR-k4-s1}}{
		\Return \textsc{No}\;
	}
	\If{$\mathcal{R}(G, I) \neq \mathcal{R}(G, J)$}{
		\Return \textsc{No}\label{algo-line:PVCR-k4-e1}\;
	}
	$G^\prime \gets G - \mathcal{R}(G, I)$\label{algo-line:PVCR-k4-se2}\;
	\ForEach{component $C$ of $G^\prime$\label{algo-line:PVCR-k4-s3}}{
		\If{$\vert V(C) \cap I \vert \neq \vert V(C) \cap J \vert$}{
			\Return \textsc{No}\;
		}
	}
	\Return \textsc{Yes}\label{algo-line:PVCR-k4-e3}\;
	
	\caption{$\mathtt{IsTSReach}(G, I, J)$.}
	\label{algo:PVCR-k4}
\end{algorithm}

The following observations are straightforward.

\begin{lemma}\label{lem:no-instances}
	Let $(G, I, J)$ be an instance of \textsc{$k$-PVCR} ($k \geq 3$) for caterpillars.
	If either $\vert I \vert \neq \vert J \vert$ or $\mathcal{R}(G, I) \neq \mathcal{R}(G, J)$, there is no $\mathsf{TS}$-sequence between $I$ and $J$.
\end{lemma}

\begin{lemma}\label{lem:remove-rigid-tokens}
	Let $(G, I, J)$ be an instance of \textsc{$k$-PVCR} ($k \geq 3$) for caterpillars.
	Suppose that $\vert I \vert = \vert J \vert$ and $\mathcal{R}(G, I) = \mathcal{R}(G, J)$.
	Then, $(G, I, J)$ is a yes-instance if and only if $(G^\prime, I \cap V(G^\prime), J \cap V(G^\prime))$ is a yes-instance, where $G^\prime$ is the graph obtained from $G$ by removing every vertex in $\mathcal{R}(G, I) = \mathcal{R}(G, J)$.
\end{lemma}

\begin{lemma}\label{lem:no-rigid-tokens}
	Let $I, J$ be two $k$-PVCs ($k \geq 4$) of a $n$-vertex caterpillar $G = (S \cup L, E)$ with $\vert I \vert = \vert J \vert = s$ and $\mathcal{R}(G, I) = \mathcal{R}(G, J) = \emptyset$.
	Then, one can construct a $\mathsf{TS}$-sequence between $I$ and $J$ in polynomial time.
\end{lemma}
\begin{proof}
	First of all, we define a total ordering $\prec$ on vertices of a caterpillar $G = (S \cup L, E)$ where $S = s_1s_2\dots s_\ell$ is the spine of $G$.
	For two distinct vertices $u, v \in V(G)$, $u \prec v$ if and only if one of the following conditions holds: 
	(a) $u = s_i$ and $v \in \bigcup_{j=i+1}^\ell L_G[s_j]$; 
	(b) $u \in L_G(s_i)$ and $v \in \bigcup_{j=i}^\ell L_G[s_j]$.
	By definition, the ordering $\prec$ between two leaves attached to the same vertex in the spine $S$ can be arbitrarily defined.
	By considering all pairs of vertices in $G$, one can construct an ordering $\prec$ between vertices in $G$ in $O(n^2)$ time.
	
	Let's fix a total ordering $\prec$ on vertices of $G$ as defined above.
	Suppose that $I = \{x_1, x_2, \dots, x_s\}$ and $J = \{y_1, y_2, \dots, y_s\}$ be two $k$-PVCs of $G$ such that $x_1 \prec x_2 \prec \dots \prec x_s$ and $y_1 \prec y_2 \prec \dots \prec y_s$.
	Our goal is to construct a $\mathsf{TS}$ between $I$ and $J$ that slides the token on $x_i$ to $y_i$, for $1 \leq i \leq s$.
	
	\begin{claim}\label{clm:construct-Si}
		Let $i \in \{1, \dots, s\}$ be such that $x_i \neq y_i$ and $x_j = y_j$ for every $j \in \{i+1, \dots, s\}$.
		Moreover, suppose that $x_i \prec y_i$.
		Then, one can construct in polynomial time a $\mathsf{TS}$-sequence $\mathcal{S}_i$ in $G$ starting from $I$ that slides the token $t_i$ on $x_i$ to $y_i$.
	\end{claim}
	\begin{proof}
		From our assumption, note that $y_i \notin I$.
		If $i = 1$, the claim is trivial: simply sliding the token on $v_i$ to $w_i$ along $P_{v_iw_i}$ is enough.
		If $x_i \in L_G(y_i)$, since every $k$-path covered by $x_i$ is also covered by $y_i$, one can immediately slide the token on $x_i$ to $y_i$.
		Thus, we can assume w.l.o.g that $2 \leq i \leq s$ and $x_i \notin L_G(y_i)$.
		
		We now claim that there must be an index $i^\star \in \{1, \dots, i\}$ such that either $x_{i^\star} \in L$ or the token on $x_{i^\star}$ can be immediately moved to $r(x_{i^\star})$.
		Since $k \geq 4$ and $\mathcal{R}(G, I) = \emptyset$, there is no token $t_x$ on $x \in I \cap S$ such that any $\mathsf{TS}$-sequence $\mathcal{S}_x$ moving $t_x$ must move it to some $y \in L_G(x)$; otherwise, one can apply $\mathcal{S}_x$ and replace the $\mathsf{TS}$-step that immediately moves $t_x$ to $y$ by a $\mathsf{TS}$-step that immediately moves $t_x$ to either $l(x)$ or $r(x)$, a contradiction.
		Suppose to the contrary that for every $j \in \{1, \dots, i\}$, we have $x_j \in S$ and the token on $x_j$ cannot be immediately moved to $r(x_j)$.
		Therefore, $\text{dist}_G(x_j, x_{j-1}) = k$ for $1 < j \leq i$, and $\text{dist}_G(s_1, x_1) = k - 2$. 
		(Recall that we assumed $\deg_G(s_1) \geq 2$.) 
		Additionally, since $J$ is a $k$-PVC, we have $\text{dist}_G(y_j, y_{j-1}) \leq k$ for $1 < j \leq i$.
		Since $x_i \in S$, $x_i \prec y_i$, and $x_i \neq y_i$, it follows that $\text{dist}_G(s_1, y_1) > \text{dist}_G(s_1, x_1) = k-2$, which contradicts the fact $J$ is a $k$-PVC, because there is some $k$-path containing $s_1$ and one of its leaf-neighbors that is not covered by any vertex in $J$.
		By simply checking tokens one by one, we can find $i^\star$ in polynomial time.
		We are interested in the largest index among all of such $i^\star$.
		
		\begin{algorithm}[!ht]
			\KwIn{A caterpillar $G = (S \cup L, E)$ on $n$ vertices, two $k$-PVCs ($k \geq 4$) $I = \{x_1, \dots, x_s\}$ and $J = \{y_1, \dots, y_s\}$ of $G$, an index $i$ such that $x_i \prec y_i$ and $x_j = y_j$ for $j \in \{i+1, \dots, s\}$.}
			\KwOut{A $\mathsf{TS}$-sequence $S_i$ that slides the token $t_i$ on $x_i$ to $y_i$.}
			\SetArgSty{textbb}  
			
			$\mathcal{S}_i \gets \langle I \rangle$; $I^\prime = \emptyset$\;
			\If{$x_i \in L_G(y_i)$}{
				$\mathcal{S}_i \gets \mathcal{S}_i \oplus \langle I, I \setminus \{x_i\} \cup \{y_i\} \rangle$\;
				\Return $\mathcal{S}_i$\;
			}
			\While{$y_i \notin I$}{
				\If{$y_i \in L$}{
					Let $y$ be the unique neighbor of $y_i$\;
					\If{$y \in I$}{
						\If{the token on $y$ cannot be immediately slid to $y_i$}{
							Find the largest index $i^\star \in \{1, \dots, i-1\}$ such that either $x_{i^\star} \in L$ or the token on $x_{i^\star}$ can be immediately slid to $r(x_{i^\star})$\;
						}
						\Else{
							$\mathcal{S}_i \gets \mathcal{S}_i \oplus \langle I,  I \setminus \{y\} \cup \{y_i\} \rangle$\; 
							\Return $\mathcal{S}_i$\;
						}
					}	
				}
				\Else{
					Find the largest index $i^\star \in \{1, \dots, i\}$ such that either $x_{i^\star} \in L$ or the token on $x_{i^\star}$ can be immediately slid to $r(x_{i^\star})$\;
				}
				\If{$x_{i^\star} \in L$}{
					\If{$N_G(x_{i^\star}) \cap I \neq \emptyset$}{
						Find the smallest index $j^\star \in \{i^\star+1, \dots, s\}$ such that $x_{j^\star} \in S$ and the token on $x_{j^\star}$ can be immediately slid to $l(x_{j^\star})$\;
						\ForEach{$j \in \{j^\star, j^\star-1, \dots, i^\star+1\}$ with $x_j \in S$}{
							$I^\prime \gets I \setminus \{x_j\} \cup \{l(x_j)\}$\;
							$\mathcal{S}_i \gets \mathcal{S}_i \oplus \langle I, I^\prime \rangle$; $I \gets I^\prime$\;
						}
					}
					$I^\prime \gets I \setminus \{x_{i^\star}\} \cup N_G(x_{i^\star})\}$\;
				}
				\Else{
					$I^\prime \gets I \setminus \{x_{i^\star}\} \cup \{r(x_{i^\star})\}$\;
				}
				
				$\mathcal{S}_i \gets \mathcal{S}_i \oplus \langle I, I^\prime \rangle$; $I \gets I^\prime$\;
			}
			\Return $\mathcal{S}_i$\;
			\caption{Construction of $S_i$.}
			\label{algo:construct-Si}
		\end{algorithm}
		
		Algorithm~\ref{algo:construct-Si} describes how to construct a $\mathsf{TS}$-sequence $S_i$ that slides the token $t_i$ on $x_i$ to $y_i$.
		Intuitively, to do this, in Algorithm~\ref{algo:construct-Si}, we repeatedly find the largest index $i^\star$ such that either $x_{i^\star} \in L$ or the token on $x_{i^\star}$ can be immediately moved to $r(x_{i^\star})$.
		If the token on $x_{i^\star}$ can be immediately moved to $r(x_{i^\star})$, we perform that token-slide.
		If $x_{i^\star} \in L$, we move the token on $x_{i^\star}$ to its unique neighbor $y$ in $S$ as follows. 
		If no token is placed on $y$, we can immediately slide the token on $x_{i^\star}$ to $y$.
		Otherwise, for $i^\star < j \leq s$ and $x_j \in S$, the token on $x_j$ cannot be immediately slid to $r(x_j)$.
		Note that from definition of $\prec$, we must have $x_{i^\star+1} = y$.
		Since $\mathcal{R}(G, I) = \emptyset$, it follows that there exists $j \in \{i^\star+1, \dots, s\}$ such that $x_j \in S$ and the token on $x_j$ can be immediately slid to $l(x_j)$, and let $j^\star$ be the smallest index among such $j$.
		Now, for $j$ from $j^\star$ downto $i^\star+1$, if $x_j \in S$, we immediately slide the token on $x_j$ to $l(x_j)$, which can be done because the distance between any two consecutive tokens in $S$ is always at most $k \geq 4$.
		Finally, we can now immediately slide the token on $x_{i^\star}$ to $y$.
		Note that, if the token on $x_i$ can be immediately moved to its right-neighbor, $i^\star = i$. 
		Then, the correctness of Algorithm~\ref{algo:construct-Si} follows.
		
		Note that, if $i^\star < i$, after $O(\text{dist}_G(v_{i^\star}, v_{i^\star + 1}))$ next iterations of the \textbf{while} loop in Algorithm~\ref{algo:construct-Si}, the value of $i^\star$ increases by at least $1$.
		On the other hand, if $i^\star = i$, the value of $i^\star$ in the next iterations may be decreased, but will finally increase again because tokens are, in some sense, moved to the ``right-hand side'' and never moved ``more than one-step'' back.
		It follows that Algorithm~\ref{algo:construct-Si} stops after $O(n)$ iterations, and therefore runs in polynomial time.
		\qed\end{proof}
	
	Using the above claim, one can construct a $\mathsf{TS}$-sequence $\mathcal{S}$ between $I = \{x_1, \dots, x_s\}$ and $J = \{y_1, \dots, y_s\}$ as follows.
	Initially, $\mathcal{S}_I = \mathcal{S}_J = \emptyset$.
	For each $i$ from $s$ downto $1$, if $x_i \prec y_i$ (resp., $y_i \prec x_i$), construct a $\mathsf{TS}$-sequence $\mathcal{S}_i^I$ (resp., $\mathcal{S}_i^J$) that slides the token on $x_i$ to $y_i$, then assign $\mathcal{S}_I$ (resp., $\mathcal{S}_J$) to $\mathcal{S}_I \oplus \mathcal{S}_i^I$ (resp., $\mathcal{S}_J \oplus \mathcal{S}_i^J$) and $I$ (resp., $J$) to the resulting $k$-PVC $I^\prime$ obtained by performing $\mathcal{S}_i^I$ (resp., $\mathcal{S}_i^J$).
	Intuitively, in each iteration, we move either a token in $I$ from $x_i$ to $y_i$ or a token in $J$ from $y_i$ to $x_i$, depending on whether $x_i \prec y_i$ or $y_i \prec x_i$.
	Finally, our desired $\mathsf{TS}$-sequence $\mathcal{S} = \mathcal{S}_I \oplus \text{rev}(\mathcal{S}_J)$.
	Since each iteration can be performed in polynomial time (by the above claim), our algorithm runs in polynomial time.
\qed\end{proof}

\section{Conclusion}\label{sec:conclusion}
In this paper, we have shown that \textsc{$k$-PVCR} for caterpillars under $\mathsf{TS}$ can be solved in polynomial time when $k \geq 4$.
Our main contribution is the characterization of ``a region surrounding a token'' which, under certain conditions, makes the token rigid.
Unfortunately, it remains unknown whether one can find all rigid tokens in polynomial time when $k = 3$, and therefore this case remains open.
We believe that this will also be useful for tackling the open question for trees.

\section*{Acknowledgements}
We thank Akira~Suzuki for his useful comments leading to some ideas presented in this paper.
This research is partially supported by the Japan Society for the Promotion of Science (JSPS) KAKENHI Grant Number JP20H05964 (AFSA).

\bibliographystyle{splncs04}
\bibliography{refs.bib}

\end{document}